\documentclass[12pt,draftcls,onecolumn]{IEEEtran}
\usepackage{amsmath}
\usepackage{amsfonts}
\usepackage{amssymb}
\usepackage[latin1]{inputenc}
\usepackage[T1]{fontenc}

\ifCLASSINFOpdf
   \usepackage[pdftex]{graphicx}
\else
   \usepackage[dvips]{graphicx}
\fi

\usepackage{array}
\usepackage{fancyhdr}
\newtheorem{remark}{Remark}
\newtheorem{corollary}{{\textit{Corollary}}}
\begin{document}

\pagestyle{plain}

\thispagestyle{fancy}
\newtheorem{proposition}{\textit{Proposition}}
\newtheorem{definition}{\textit{Definition}}
\newtheorem{problem}{\textit{Problem}}

%
\renewcommand{\headrulewidth}{0pt} \fancyhf{} This work has been
submitted to the IEEE for possible publication. Copyright may be
transferred without notice, after which this version may no longer
be accessible.

\title{Delay-limited Source and Channel Coding of Quasi-Stationary Sources over Block Fading Channels: Design and Scaling Laws}

\author{\IEEEauthorblockN{Roghayeh Joda$^{\dag*}$ and Farshad~Lahouti$^{\dag}$}\\
\IEEEauthorblockA{$\dag$School of Electrical and Computer Engineering, University of Tehran, Iran\\
$^*$Department of Electrical and Computer Engineering, Polytechnic Institute of New York University, USA\\
Email: [rjoda, lahouti]@ut.ac.ir, http://wmc.ut.ac.ir}
 }
\maketitle \thispagestyle{empty}
 \pagestyle{empty}
\begin{abstract}
In this paper, delay-limited transmission of quasi-stationary
sources over block fading channels are considered. Considering
distortion outage probability as the performance measure, two source
and channel coding schemes with power adaptive transmission are
presented. The first one is optimized for fixed rate transmission,
and hence enjoys simplicity of implementation. The second one is a
high performance scheme, which also benefits from optimized rate
adaptation with respect to source and channel states. In high SNR
regime, the performance scaling laws in terms of outage distortion
exponent and asymptotic outage distortion gain are derived, where
two schemes with fixed transmission power and adaptive or optimized
fixed rates are considered as benchmarks for comparisons. Various
analytical and numerical results are provided which demonstrate a
superior performance for source and channel optimized rate and power
adaptive scheme. It is also observed that from a distortion outage
perspective, the fixed rate adaptive power scheme substantially
outperforms an adaptive rate fixed power scheme for delay-limited
transmission of quasi-stationary sources over wireless block fading
channels. The effect of the characteristics of the quasi-stationary
source on performance, and the implication of the results for
transmission of stationary sources are also investigated.
\end{abstract}
\IEEEpeerreviewmaketitle
\pagestyle{plain} \setcounter{page}{0}
\section{Introduction}\label{SI}
Multimedia signals such as speech and video are usually
quasi-stationary and their transmission in real-time or streaming
applications is subject to certain delay constraints. The delay
limited communications over a wireless block fading channel is
studied from a channel coding perspective in, e.g.,
\cite{{R1},{R5}}, where the performance is quantified in terms of
channel outage probability, outage capacity and delay-limited
capacity. In this paper, we study the delay-limited transmission of a
quasi-stationary source over a block fading channel from the
perspective of source and channel coding designs and performance
scaling laws.

The zero outage capacity region of the multiple access and the
broadcast block fading channels, respectively are studied in
\cite{R5} and \cite{R4}. In \cite{R5}, it is shown that the delay
limited capacity of a single user Rayleigh block fading channel is
zero. In \cite{R1}, \cite{R4} and \cite{R6}, power adaptation for
constant rate transmission over point-to-point, broadcast and
multiple access block fading channels is designed for minimum outage
probability. The outage performance of the relay block fading
channel is investigated in, e.g., \cite{R7}\nocite{R8}-\cite{R9}.

The end to end mean distortion for transmission of a stationary
source over a block fading channel is considered in
\cite{R11}\nocite{R14,R15,R16,R17}-\cite{R18}. The performance is
studied in terms of the (mean) distortion exponent or the decay rate
of the end to end mean distortion with (channel) signal to noise
ratio (SNR) in high SNR regime.

The transmission of a stationary source over a MIMO block fading
channel is considered in \cite{R21}, where the distortion outage
probability and the outage distortion exponent are considered as
performance measures. For constant power transmission, it is shown
in \cite{R21} that separate source and channel coding schemes with
constant (optimized) or adaptive transmission rate essentially
provide the same distortion outage probability.

We consider the delay-limited transmission of a quasi-stationary
source over a wireless block fading channel. The assumption is that
the channel state information is known at the transmitter. The
source and channel separation does not hold in this setting
\cite{R19}\cite{R20}, however, for practical reasons we are
interested in exploring the designs that combine conventional high
performance source codes and channel codes in an optimized manner.
Specifically, a framework for rate and/or power adaptation using
source and channel codes, that achieve the rate-distortion and the
capacity in a given state of source and channel, is presented. The
applicable performance measures of interest, as described in Section
II, are the probability of distortion outage and the outage
distortion exponent. Under an average transmission power constraint,
two designs are presented. The first scheme devises a channel
optimized power adaptation to minimize the distortion outage
probability for a given optimized fixed rate, and hence enjoys the
simplicity of single rate transmission. The second scheme formulates
adaptation solutions for transmission power and source and channel
coding rate such that the distortion outage probability is
minimized.  As benchmarks, we consider two constant power
delay-limited communication schemes with channel optimized adaptive
or fixed rates.

The performance of the presented schemes are assessed and compared
both analytically and numerically. Specifically for large enough
SNR, different scaling laws involving outage distortion exponent and
asymptotic outage distortion gain are derived. The analyses are
mainly derived for wireless block fading channels and are
specialized to Rayleigh block fading channels in certain cases. The
results demonstrate the superior performance of the source and
channel optimized rate and power adaptive scheme. An interesting
observation is that from a distortion outage perspective, an
adaptive power single rate scheme noticeably outperforms a rate
adaptive scheme with constant transmission power. This is the
opposite of the observation made in \cite{R22} from the Shannon
capacity perspective. The effect of the statistics of
quasi-stationary source on the performance of the presented schemes
is also investigated. In the marginal case of a stationary source,
our studies reveal that a fixed optimized rate provides the same
outage distortion as the optimized rate adaptation scheme, either
with adaptive or constant transmission power. The results shed light
on proper cross-layer design strategies for efficient and reliable
transmission of quasi-stationary sources over block fading channels.

The paper is organized as follows. Following the preliminaries and
the description of system model in Section \ref{SII}, Section
\ref{SIII} presents the design based on fixed source coding rate for
minimized distortion outage probability. Next, in Section \ref{SIV},
we present the adaptive rate and power source and channel coding
design. Finally performance evaluations and comparisons are
presented in Section \ref{SV}.

\section{System Model}\label{SII}
We consider the transmission of a quasi-stationary source over a
block fading channel. Specifically, the source is finite state
quasi-stationary Gaussian with zero mean and variance $\sigma^2_s$
in a given block, where
$s\in\mathcal{S}:\;\mathcal{S}=\{1,2,...,N_s\}$\cite{R2}. The source
state $s$ from the set $\mathcal{S}$ is a discrete random variable
with the probability density function (pdf) $P(s)$. The source
coding rate in a block in state $s$, is denoted by $R_s$ bits per
source sample. Hence, according to the distortion-rate function of a
Gaussian source \cite{R2}\cite{R3}, the instantaneous distortion in
a block in state $s$ is given by $D=\sigma^2_s\,2^{-2R_s}$.

We consider a point to point wireless block fading channel for
transmitting the source information to the destination. Let $X$, $Y$
and $Z$, respectively indicate channel input, output and additive
noise, where $Z$ is an i.i.d Gaussian noise $\sim\mathcal{N}(0,1)$.
Therefore, we have $Y=\alpha X+Z$, where $\alpha$ is the
multiplicative fading which is constant across one block and
independently varies from one block to another according to the
continuous probability density function $f(\alpha)$.
For a Rayleigh fading channel, the channel gain $|\alpha|^2$ is an
exponentially distributed random variable, where we here consider
$E[\alpha]=1$.

The block diagram of the system is depicted in Fig. \ref{F20}. We
consider $K$ source samples spanning one source block coded into a
finite index by the source encoder. This index is transmitted in $N$
channel uses spanning one fading block (bandwidth expansion ratio
$b=\frac{N}{K}$, here $b\in\mathbb{N}$). We assume that $K$ and $N$
are large enough such that, over a given state of source and
channel, the rate distortion function of the quasi-stationary source
and the instantaneous capacity of the block fading channel may be
achieved.
The source coding rate $R_s$ in bits per source sample and channel
coding rate $R$ in bits per channel use are related by $R_s=bR$. The
instantaneous capacity of the fading Gaussian channel \cite{R1} over
one block (in bits per channel use) is defined as \vspace{-4pt}
\begin{equation}\label{E40}
C(\alpha,\gamma)=\frac{1}{2}\log_2(1+\alpha\gamma){.}
\end{equation}\vspace{-2pt}
In case of a channel outage, in each state of the source and the
channel $(s,\alpha)$, the instantaneous distortion is equal to the
variance of the source and the decoder reconstructs the mean of the
source. Whereas without channel outage, distortion is given by
$\sigma^2_s2^{-2bR}$. Thus, the distortion at a given state
$(s,\alpha)$ is equal to\vspace{-8pt}
\begin{equation}\label{E10}
D{(\sigma_s,\alpha,\gamma)}=\begin{cases}
\sigma^2_s 2^{-2bR} &\text{if}\; R\leq C(\alpha,\gamma) \\
\sigma^2_s   &\text{if}\; R>C(\alpha,\gamma){,}
\end{cases}
\end{equation}
where $C(\alpha,\gamma)$ is given in \eqref{E40}. Let $D_m$ be a
nonnegative constant and represent the maximum allowable distortion.
The distortion outage probability evaluated at $D_m$ is defined as
\begin{equation}\label{E43}
\text{P}_\text{Dout}=\text{Pr}(D(\sigma_s,\alpha,\gamma)> D_m){,}
\end{equation}
where $\gamma=\gamma(\sigma_s,\alpha)$ is the transmission power and
we have the power constraint $\text{E}[\gamma]\leq \bar{P}$.

The outage distortion exponent is defined as \cite{R21}\vspace{-2pt}
\begin{equation}\label{E83}
\Delta_{\textit{OD}}=\underset{\bar{P}\rightarrow\infty}{\lim}{-\frac{\ln{\text{P}_\text{Dout}}}{\ln{\bar{P}}}}.
\end{equation}\vspace{-2pt}
Let $\bar{P}_1$ and $\bar{P}_2$ be the average powers transmitted to
asymptotically achieve a specific distortion outage probability by
two different schemes. We define the asymptotic outage distortion
gain as follows\vspace{-10pt}
\begin{equation}\label{E68}
G\!_{O\!D}=10\log_{10}{\bar{P}_2}-10\log_{10}{\bar{P}_1}.
\end{equation}\vspace{-2pt}
In the sequel, we also use the following mathematical
definitions\vspace{-4pt}
\begin{equation}\label{E94}
  [x]^+:=\operatorname{max}\{x,0\}
\end{equation}
\vspace{-4pt}and\vspace{-4pt}
\begin{equation}\label{E95}
E_1(x):={\int^\infty_x}{\frac{e^{-\alpha}}{\alpha}\mathrm{d}\alpha},
\end{equation}
as well as the following two approximations (refer to (5.1.53) in
\cite{R13})\vspace{-2pt}
\begin{equation}\label{E98}
E_1(x)\cong-E_c-\ln(x) \quad \text{if} \;x\rightarrow 0\;,
\end{equation}\vspace{-2pt}
where $E_c=0.5772156649$ is the Euler constant and\vspace{-2pt}
\begin{equation}\label{E99}
e^x\cong1+x \quad \text{if} \;x\rightarrow 0.
\end{equation}

\section{Channel Optimized Power Adaptation with Fixed Rate Source and Channel Coding}\label{SIII}
In this section, the aim is to find the optimized power allocation
strategy and fixed rate such that the distortion outage probability
for communication of a quasi-stationary source over a wireless
fading channel is minimized. With a fixed rate, the encoders do not
need to be rate adaptive which simplifies the design and
implementation of transceivers. Noting \eqref{E43} the distortion
outage probability is computed as follows\vspace{-4pt}
\begin{align}\label{E75}
\text{P}_\text{Dout}=\text{Pr}\bigl(R>
C\left(\alpha,\gamma\right)\bigr)\text{Pr}\bigl(\sigma^2_s>D_m\bigr)+\bigl(1-\text{Pr}\bigl(R>
C\left(\alpha,\gamma\right)\bigr)\bigr)\text{Pr}\bigl(\sigma^2_s2^{-2bR}>D_m\bigr).
\end{align}
We have the following design problem.
\begin{problem}\label{PR2}
The problem of delay-limited channel optimized power adaptation for
communication of a quasi-stationary source with minimum distortion
outage probability (COPA-MDO) is formulated as follows\vspace{-8pt}
\begin{align}\label{E76}
&\underset{\gamma,R}{\operatorname{min}}\;\text{P}_\text{Dout} \nonumber\\
&\text{subject to}\; \text{E}[\gamma]\leq\bar{P}{.}
\end{align}
\end{problem}
The solution to Problem \ref{PR2} is obtained in two steps. For a
given fixed rate and noting
$\text{Pr}\bigl(\sigma^2_s2^{-2bR}>D_m\bigr)\leq\text{Pr}\bigl(\sigma^2_s>D_m\bigr):\forall
R>0$, minimizing \eqref{E75} is equivalent to minimizing
$\text{Pr}\bigl(R>C\left(\alpha,\gamma\right)\bigr)$. Problem
\ref{PR3} below formulates this (sub)problem and provides the
optimum power adaptation strategy as a function of $R$. Next,
solving Problem \ref{PR6}, as described below, provides the optimum
solution for $R$, and hence problem \ref{PR2} is solved.

\begin{problem}\label{PR3}
With COPA-MDO scheme and with a given fixed channel coding rate $R$,
the power adaptation problem is formulated as follows\vspace{-4pt}
\begin{align}\label{E21}
\begin{split}
\underset{\gamma}{\operatorname{min}}\;\text{Pr}(R>C(\alpha,\gamma)) \\
\text{subject to}\; \text{E}[\gamma]\leq\bar{P}{.}
\end{split}
\end{align}
\end{problem}

\begin{proposition}\label{P3}
The solution to Problem \ref{PR3} for optimized power adaptation
over a block fading channel is given by
\begin{equation}\label{E28}
\gamma^*{(\alpha,R)}=\begin{cases}
\frac{2^{2R}-1}{\alpha} &\text{if}\;\;\alpha\geq\frac{2^{2R}-1}{q^*_1} \\
0 &\text{if}\;\;\alpha<\frac{2^{2R}-1}{q^*_1}{,}
\end{cases}
\end{equation}
in which $q^*_1$ is selected such that the power constraint is
satisfied with equality, i.e.,
\begin{equation}\label{E129}
\underset{\alpha\geq\frac{2^{2R}-1}{q^*_1}}{\int}{\frac{2^{2R}-1}{\alpha}
f(\alpha)\,\mathrm{d}\alpha}=\bar{P}.
\end{equation}
Specifically for Rayleigh block fading channel, \eqref{E129} is
simplified to
\begin{equation}\label{E54}
\left(2^{2R}-1\right)E_1\left(\frac{2^{2R}-1}{q^*_1}\right)=\bar{P}.
\end{equation}
\end{proposition}
\begin{proof}
The proof follows that of Proposition 4 in \cite{R1}. Based on
\eqref{E129}, for Rayleigh block fading channel $q_1^*$ is to
satisfy \eqref{E54}. In fact $q_1^*$ sets the SNR threshold below
which the channel outage occurs.
\end{proof}
Now for obtained $\gamma^*(\alpha,R)$, minimizing \eqref{E75} is
equivalent to maximizing
\begin{equation}\label{E23}
\Lambda(R)=\text{Pr}\biggl(R\leq
C\left(\alpha,\gamma^*(\alpha,R)\right)\biggr)\biggl(\text{Pr}\left(\sigma^2_s>D_m\right)-\text{Pr}\left(\sigma^2_s2^{-2bR}>D_m\right)\biggr),
\end{equation}
as formulated in Problem \ref{PR6} below.
\begin{problem}\label{PR6}
For COPA-MDO scheme with optimum power adaptation
$\gamma^*(\alpha,R)$, the optimum $R$, denoted by $R^*$, is given by
the solution to the following optimization problem
\begin{equation}
\begin{split}
&\underset{R}{\operatorname{max}}\;\Lambda(R)\\
&\text{subject
to}\;\text{E}_{\alpha}[\gamma^*(\alpha,R)]\leq\bar{P}.
\end{split}
\end{equation}
\end{problem}

\begin{proposition}\label{P4}
Solution to Problem \ref{PR6} for block fading channel is $R^*$,
where $R^*$ is an element of the set
$\{\frac{1}{2b}\log_2\frac{\sigma^2_s}{D_m};s\in\mathcal{S},\;\sigma^2_s\geq
D_m\}$ that maximizes
\begin{equation}
{\text{Pr}\left(\alpha\geq\frac{2^{2R}-1}{q^*_1}\right)\biggl(\text{Pr}\left(\!\sigma^2_s>D_m\right)-\text{Pr}\left(\sigma^2_s2^{-2bR}>D_m\right)\!\biggr)}\!
\end{equation}
for $q_1^*$ satisfying \eqref{E129}.
\end{proposition}
%
\begin{proof}
Using \eqref{E40} and \eqref{E28}, we have
\begin{equation}\label{E37}
C(\alpha,\gamma)=\begin{cases}
R &\text{if}\;\frac{2^{2R}-1}{\alpha}\leq q^*_1\\
0 &\text{if}\;\frac{2^{2R}-1}{\alpha}> q^*_1
\end{cases}
\end{equation}
Therefore,
\begin{align}\label{E41} \text{Pr}\left(R\leq
C\left(\alpha,\gamma\right)\right)=\text{Pr}\left(\alpha\geq\frac{2^{2R}-1}{q^*_1}\right)
\end{align}
Thus, the maximization in Problem \ref{PR6} is equivalent to
\begin{equation}\label{E50}
\underset{R}{\operatorname{max}}\;{\text{Pr}\left(\alpha\geq\frac{2^{2R}-1}{q^*_1}\right)\biggl(\!\text{Pr}\left(\sigma^2_s>D_m\right)\!-\!\text{Pr}\left(\sigma^2_s2^{-2bR}>D_m\right)\!\biggr)}\\
\end{equation}
where $q_1^*$ satisfies \eqref{E129}.

It is clear in \eqref{E50} that
$\text{Pr}\left(\sigma^2_s>D_m\right)-\text{Pr}\left(\sigma^2_s2^{-2bR}>D_m\right)$
is an ascending (stair case like) function of $R$ with discontinuity
at $R=\frac{1}{2b}\log_2\frac{\sigma^2_{s}}{D_m}; \forall
s\in{\mathcal{S}},\;\sigma^2_s\geq D_m$. On the other hand, noting
\eqref{E129}, $\frac{\left(2^{2R}-1\right)}{q^*_1}$ increases as $R$
grows; and therefore
$\text{Pr}\left(\alpha\geq\frac{2^{2R}-1}{q^*_1}\right)$ decreases
as $R$ increases. Thus, the objective function in Problem \ref{PR6}
has discontinuities at
$R=\frac{1}{2b}\log_2\frac{\sigma^2_{s}}{D_m}; \forall
s\in{\mathcal{S}},\;\sigma^2_s\geq D_m$ and also is a descending
function of $R$ between two subsequent discontinuity points. Hence,
we can conclude that the maximum of \eqref{E50} occurs at one of
these discontinuity points.
\end{proof}

The next two Propositions quantify the performance of the proposed
COPA-MDO scheme in terms of the resulting distortion outage
probability and outage distortion exponent, respectively.
\begin{proposition}\label{P9}
The distortion outage probability obtained by COPA-MDO scheme for
transmission of a quasi-stationary source over a Rayleigh block
fading channel is given by
\begin{align}\label{E57}
\text{P}_\text{Dout}=\left(1-\exp\left(-\frac{2^{2R^*}-1}{q^*_1}\right)\right)\text{Pr}(\sigma^2_s>D_m)
+\exp\left(-\frac{2^{2R^*}-1}{q^*_1}\right)\text{Pr}(\sigma^2_s2^{-2bR^*}>D_m)
\end{align}
with $q_1^*$ satisfying the following equation
\begin{equation}\label{E59}
\left(2^{2R^*}-1\right)E_1\left(\frac{2^{2R^*}-1}{q^*_1}\right)=\bar{P}
\end{equation}
 and $R^*$ given in Proposition \ref{P4}.
\end{proposition}

\begin{proof}
Noting \eqref{E41} and with exponentially distributed channel gain,
we obtain
\begin{align}\label{E33} \text{Pr}\left(R^*\leq
C\left(\alpha,\gamma\right)\right)=\exp\left(-\frac{2^{2R^*}-1}{q^*_1}\right),
\end{align}
where $R^*$ is given in Proposition \ref{P4}. Hence, equation
\eqref{E57} is derived by replacing \eqref{E33} in \eqref{E75}.
\end{proof}
As expected, the power allocation and hence the resulting distortion outage
probability are functions of the optimized fixed rate.

\begin{proposition}\label{P13}
For communication of a quasi-stationary source over a Rayleigh block
fading channel, the COPA-MDO scheme achieves the outage distortion
exponent $\Delta_{\textit{OD}}$ of the order
$O\left(\frac{\bar{P}}{\ln{\bar{P}}}\right)$ for large average power
limit $\bar{P}$.
\end{proposition}
\begin{proof}
To compute $\Delta_{\textit{OD}}$ as defined in \eqref{E83}, we
first note that according to \eqref{E59} and Proposition \ref{P4}
for large enough $\bar{P}$, we have
$\frac{2^{2R^*}-1}{q_1^*}\rightarrow 0$. Thus, using the
approximation given in \eqref{E98}, \eqref{E59} can be rewritten as
follows.
\begin{equation}
-\left(2^{2R^*}-1\right)\left(\ln\frac{2^{2R^*}-1}{q_1^*}+E_c\right)\cong\bar{P}
\end{equation}
or equivalently
\begin{equation}\label{E100}
q_1^*\cong\left(2^{2R^*}-1\right)\exp\left(\frac{\bar{P}}{2^{2R^*}-1}\right).
\end{equation}
Therefore, using \eqref{E99} and \eqref{E100} and noting
$\frac{2^{2R^*}-1}{q_1^*}\rightarrow 0$, \eqref{E57} can be
rewritten as follows

\begin{equation}\label{E117}
\begin{split}
\text{P}_\text{Dout}\cong
\text{Pr}\left(\sigma^2_s>D_m\right)\exp\left(-\frac{\bar{P}}{2^{2R^*}-1}\right)+\left(1-\exp\left(-\frac{\bar{P}}{2^{2R^*}-1}\right)\right)\text{Pr}(\sigma^2_s2^{-2bR^*}>D_m).
\end{split}
\end{equation}
Noting Proposition \ref{P4} or equivalently minimizing \eqref{E117}
for large power constraint, it is observed that the optimum $R^*$
have to be set such that $\text{Pr}(\sigma^2_s2^{-2bR^*}>D_m)=0$,
i.e.,
\begin{equation}
R^*=\frac{1}{2b}\log_2\frac{\underset{s}{\operatorname{max}}\{\sigma^2_s\}}{D_m}
\end{equation}
Thus, we have
\begin{equation}\label{E69}
\begin{split}
\text{P}_\text{Dout}\cong
\text{Pr}\left(\sigma^2_s>D_m\right)\exp\left(\frac{-\bar{P}}{\left(\frac{\underset{s}{\operatorname{max}}\{\sigma^2_s\}}{D_m}\right)^{\frac{1}{b}}-1}\right)
\end{split}
\end{equation}
and therefore,\vspace{-4pt}
\begin{equation}\label{E26}
\begin{split}
\Delta_{\textit{OD}}=\underset{\bar{P}\rightarrow\infty}{\lim}{\frac{\frac{\bar{P}}{\left(\frac{\underset{s}{\operatorname{max}}\{\sigma^2_s\}}{D_m}\right)^{\frac{1}{b}}-1}-\ln{\text{Pr}\left(\sigma^2_s>D_m\right)}}{\ln{\bar{P}}}}.
\end{split}
\end{equation}
Hence, the proof is complete.
\end{proof}
For the optimized fixed rate $R^*$, the distortion outage exponent enhances when the average power limit $\bar{P}$ increases.

In the following three Corollaries, we summarize the implications of
the COPA-MDO design for transmission of stationary sources over
block fading channels. The stationary source is a Gaussian with zero
mean and variance $\sigma^2\geq D_m$. Obviously, with
$\sigma^2<D_m$, the distortion outage probability is equal to zero.
The results are directly obtained from Propositions \ref{P3},
\ref{P4}, \ref{P9} and \ref{P13} and allows for insights into the
system performance as it relates to source statistical
characteristics.
\begin{corollary}\label{P10}
The optimum power adaptation and channel coding rate prescribed by
COPA-MDO for transmission of a stationary source over a block fading
channel are given by
\begin{align}\label{E8}
R^*=\frac{1}{2b}\log_2{\frac{\sigma^2}{D_m}}
\end{align}
and\vspace{-4pt}
\begin{align}\label{E9}
\gamma^*=\begin{cases}\frac{\left(\frac{\sigma^2}{D_m}\right)^\frac{1}{b}-1}{q^*_1}&\text{if
}\alpha\geq\frac{\left(\frac{\sigma^2}{D_m}\right)^\frac{1}{b}-1}{q^*_1}
\\0&\text{otherwise}\end{cases}
\end{align}
in which $q^*_1$ is selected such that the power constraint is
satisfied with equality, i.e.,
\begin{equation}\label{E3}
{\int_{\alpha\geq\frac{\left(\frac{\sigma^2}{D_m}\right)^\frac{1}{b}-1}{q^*_1}}}{\frac{\left(\frac{\sigma^2}{D_m}\right)^\frac{1}{b}-1}{\alpha}
f(\alpha)\,\mathrm{d}\alpha}=\bar{P}.
\end{equation}\vspace{-6pt}
Specifically for Rayleigh block fading channel, \eqref{E3} is
simplified to
\begin{equation}\label{E7}
\left(\left(\frac{\sigma^2}{D_m}\right)^\frac{1}{b}-1\right)E_1\left(\frac{\left(\frac{\sigma^2}{D_m}\right)^\frac{1}{b}-1}{q^*_1}\right)=\bar{P}.
\end{equation}
\end{corollary}\vspace{-2pt}
Thus, for transmission of a stationary source over a block fading
channel, the optimized fixed rate is simply set such that the source
coding distortion is equal to its maximum $D_m$.
\begin{corollary}\label{P19}
The distortion outage probability obtained by COPA-MDO for
transmission of a stationary source over a Rayleigh block fading
channel is given by\vspace{-2pt}
\begin{align}\label{E51}
\text{P}_\text{Dout}=1-\exp\left(-\frac{\left(\frac{\sigma^2}{D_m}\right)^\frac{1}{b}-1}{q^*}\right)
\end{align}
with $q^*$ satisfying the following equation\vspace{-4pt}
\begin{equation}\label{E34}
\left(\left(\frac{\sigma^2}{D_m}\right)^\frac{1}{b}-1\right)E_1\left(\frac{\left(\frac{\sigma^2}{D_m}\right)^\frac{1}{b}-1}{q^*}\right)=\bar{P}.
\end{equation}
\end{corollary}
The results in Corollary \ref{P19} is obtained noting that
$\sigma^2>D_m$ and $\text{Pr}(\sigma^2_s2^{-2bR^*}>D_m)=0$ for
stationary sources.

\begin{corollary}\label{P21}
For communication of a stationary source over a Rayleigh block
fading channel and with large average power limit $\bar{P}$, the
COPA-MDO scheme achieves the outage distortion exponent
$\Delta_{\textit{OD}}=\underset{\bar{P}\rightarrow\infty}{\lim}{\frac{\bar{P}}{\ln\bar{P}\left(\left(\frac{\sigma^2}{D_m}\right)^{\frac{1}{b}}-1\right)}}$
of the order $O\left(\frac{\bar{P}}{\ln{\bar{P}}}\right)$.
\end{corollary}
This immediately follows the proof of Proposition \ref{P13} and noting
that for stationary sources $R^*=\frac{1}{2b}\log_2{\frac{\sigma^2}{D_m}}$.

The performance of COPA-MDO scheme is studied and compared in
Section \ref{SV}.
\section{Source and Channel Optimized Power and Rate Adaptation}\label{SIV}
In this section, we consider power and rate adaptation with regard
to source and channel states for improved performance of
communications of a quasi-stationary source over a wireless block
fading channel. Thus, the objective in this section is to devise
power and rate adaptation strategies for each state $(s,\alpha)$
such that the distortion outage probability is minimized, when the
average power is constrained to $\bar{P}$. We have the following
design problem.
\begin{problem}\label{PR7}
The problem of delay-limited source and channel optimized power
adaptation for transmission of a quasi-stationary source with
minimum distortion outage probability (SCOPA-MDO) over a block
fading channel is formulated as follows
\begin{equation}\label{E48}
\begin{split}
&\underset{\gamma,R}{\operatorname{min}}\;\text{P}_\text{Dout}=\text{Pr}\left(D(\sigma_s,\alpha,\gamma)>D_m\right)\\
&\text{subject to}\; \text{E}[\gamma]<\bar{P}{,}
\end{split}
\end{equation}
where $D(\sigma_s,\alpha,\gamma)$ is given in \eqref{E10}.
\end{problem}

\begin{proposition}\label{P7}
The solution to Problem \ref{PR7} for an arbitrary block fading
channel is given by
\begin{align}\label{E14}
\gamma^*{(\sigma_s,\alpha)}=\begin{cases}
\frac{1}{\alpha}\left[\left(\frac{\sigma^2_s}{D_m}\right)^{\frac{1}{b}}-1\right]^+
&\text{if}\;\;\frac{1}{\alpha}\left[\left(\frac{\sigma^2_s}{D_m}\right)^{\frac{1}{b}}-1\right]^+\leq q^*_2\\
0 &\text{Otherwise}\\
\end{cases}
\end{align}\vspace{-2pt}
and
\begin{equation}\label{E35}
\begin{split}
R=&\begin{cases} \frac{1}{2b} \log\frac{\sigma^2_s}{D_m} &\text{if}
\;\frac{\sigma^2_s}{D_m}> 1
\;\text{and}\;\frac{1}{\alpha}\left[\left(\frac{\sigma^2_s}{D_m}\right)^{\frac{1}{b}}-1\right]<
q^*_2\\
0 &\text{Otherwise},
\end{cases}
\end{split}
\end{equation}
with $q^*_2$ satisfying the following equation
\begin{equation}\label{E19}
\underset{s:\sigma^2_s>D_m}{\sum}\underset{\alpha:\frac{1}{\alpha}\left[\left(\frac{\sigma^2_s}{D_m}\right)^{\frac{1}{b}}-1\right]\leq
q^*_2}{\int}{\frac{\left(\frac{\sigma^2_s}{D_m}\right)^{\frac{1}{b}}-1}{\alpha}
f(\alpha)\,P(s)\mathrm{d}\alpha}=\bar{P}.
\end{equation}
\end{proposition}
\begin{proof}
The transmission rate may be controlled with respect to the channel
state. Specifically, it is logical to set the rate to its maximum as
follows
\begin{equation}\label{E12}
R=C\left(\alpha,\gamma\right)=\frac{1}{2}\log(1+\alpha\gamma)
\end{equation}
and therefore, we have
\begin{align}\label{E24}
D(\sigma_s,\alpha,\gamma)=\sigma^2_s
2^{-2bR}=\frac{\sigma^2_s}{(1+\alpha\gamma)^b}.
\end{align}
Let $\gamma$ represent the probabilistically adapted transmission
power. Obviously with distortion outage, $\gamma$ is set to zero and
without distortion outage, the power is set to $\acute{\gamma}$ such
that $D(\sigma_s,\alpha,\acute\gamma)\leq D_m$. We define
$\omega(\sigma_s,\alpha)$ to be the probability that (noting the
channel SNR) the power is not zero in state $\sigma_s$ and $\alpha$.
We have,
\begin{equation}\label{E39}
\gamma(\sigma_s,\alpha)=\begin{cases} \acute\gamma(\sigma_s,\alpha)
&\text{with probability}
\;\omega(\sigma_s,\alpha)\\
0 &\text{with probability}\; 1-\omega(\sigma_s,\alpha).
\end{cases}
\end{equation}
Thus, the distortion outage probability and the average transmitted
power, respectively are $1-E[\omega(\sigma_s,\alpha)]$ and
$E[\omega(\sigma_s,\alpha)\acute\gamma(\sigma_s,\alpha)]$. Hence,
using a Lagrangian optimization approach, the problem may be
restated as follows in which the Lagrangian variable $\lambda$ is
set such that the power constraint in \eqref{E48} is satisfied with
equality.
\begin{align}\label{E25}
\underset{{\omega(\sigma_s,\alpha),\acute\gamma(\sigma_s,\alpha)
}}{\operatorname{max}}&{E[\omega(\sigma_s,\alpha)]-\lambda
E[\omega(\sigma_s,\alpha)\acute\gamma(\sigma_s,\alpha)]}\quad\text{subject to}\\
&{C_1:
\frac{\sigma^2_s}{(1+\alpha\acute\gamma(\sigma_s,\alpha))^b}\leq
D_m}\nonumber\\
&{C_2:
 0<\omega(\sigma_s,\alpha)\leq1}\nonumber.
\end{align}
Note that $\omega(\sigma_s,\alpha)=0$ is a trivial case that is
excluded in $C_2$ above. Due to the fact that $C_1$ (representing
$D(\sigma_s,\alpha,\acute\gamma)\leq D_m$) and $C_2$ are convex, the
solution to the above for $\lambda>0$ is globally optimal. To solve
this problem, we first find ${\gamma}^*(\sigma_s,\alpha)$ to
maximize \eqref{E25} for a given $\omega(\sigma_s,\alpha)$. Then for
the resulting power adaptation strategy, we obtain
${\omega}^*(\sigma_s,\alpha)$ to maximize \eqref{E25}. Thus,
${\gamma}^*(\sigma_s,\alpha)$ is the solution to
\begin{equation}\label{E13}
\begin{split}
\underset{{\acute\gamma(\sigma_s,\alpha)}
}{\operatorname{min}}&{\lambda
E[\omega(\sigma_s,\alpha)\acute\gamma(\sigma_s,\alpha)]}\quad\text{subject to}\\
&C_1:
\frac{\sigma^2_s}{(1+\alpha\acute\gamma(\sigma_s,\alpha))^b}\leq
D_m.
\end{split}
\end{equation}
Noting $\lambda>0$ and $0<\omega(\sigma_s,\alpha)\leq1$, \eqref{E13}
is equivalent to the following
\begin{equation}\label{E27}
\begin{split}
&\operatorname{min}{\acute\gamma(\sigma_s,\alpha)} \\
&\text{subject to}\;
\frac{\sigma^2_s}{(1+\alpha\acute\gamma(\sigma_s,\alpha))^b} \leq
D_m.
\end{split}
\end{equation}
Obviously, as ${\gamma}^*(\sigma_s,\alpha)$, the solution to the
problem above, is to be nonnegative, we have
\begin{equation}\label{E16}
{\gamma}^*(\sigma_s,\alpha)=\frac{1}{\alpha}\left[\left(\frac{\sigma^2_s}{D_m}\right)^{\frac{1}{b}}-1\right]^+.
\end{equation}
Now, given the power adaptation ${\gamma}^*(\sigma_s,\alpha)$, the
optimization problem \eqref{E25} may be rewritten as follow
\begin{equation}\label{E31}
\underset{0\leq\omega(\sigma_s,\alpha)\leq1}{\operatorname{max}}{\text{E}\left[\omega(\sigma_s,\alpha)\left(1-\lambda
{\gamma}^*(\sigma_s,\alpha)\right)\right]}.
\end{equation}
The solution to \eqref{E31} is given by
\begin{equation}\label{E22}
{\omega}^*(\sigma_s,\alpha)=\begin{cases}
1 &\text{if}\;\;{\gamma}^*{(\sigma_s,\alpha)}\leq\frac{1}{\lambda} \\
0 &\text{if}\;\;{\gamma}^*{(\sigma_s,\alpha)}>\frac{1}{\lambda}{.}
\end{cases}
\end{equation}
Therefore, noting \eqref{E16}, \eqref{E22} and \eqref{E39}, we
obtain \eqref{E14}. Substituting \eqref{E14} in \eqref{E12},
achieving \eqref{E35} is straightforward.

The Lagrangian variable $\lambda$ is set such that the power
constraint
$\text{E}[{\omega}^*(\sigma_s,\alpha){\gamma}^*(\sigma_s,\alpha)]=\bar{P}$
is satisfied. We have
\begin{equation}\label{E20}
\underset{s\in{\mathcal{S}}}{\sum}\underset{\alpha:{\gamma}^*(\sigma_s,\alpha)\leq
\frac{1}{\lambda}}{\int}{{\gamma}^*(\sigma_s,\alpha)
f(\alpha)\,P(s)\mathrm{d}\alpha}=\bar{P},
\end{equation}
where replacing $\frac{1}{\lambda}$ with $q^*_2$ completes the
proof.
\end{proof}

The next two Propositions quantify the performance of the proposed
SCOPA-MDO scheme in terms of the resulting distortion outage
probability and outage distortion exponent, respectively.
\begin{proposition}\label{P8}
The distortion outage probability for transmission of a
quasi-stationary source using the SCOPA-MDO scheme over a Rayleigh
block fading channel is given by\vspace{-2pt}
\begin{equation}\label{E46}
\text{P}_\text{Dout}=\text{Pr}(\sigma^2_s>D_m)\!-\!\underset{s:\sigma^2_s>D_m}\sum\!\exp{\left(-\frac{\left(\frac{\sigma^2_s}{D_m}\right)^{\frac{1}{b}}\!-\!1}{q^*_2}\right)}P(s)
\end{equation}
with $q^*_2$ satisfying the following equation\vspace{-4pt}
\begin{equation}\label{E102}
\begin{split}
\underset{s:\sigma^2_s>
D_m}{\sum}{\left(\left(\frac{\sigma^2_s}{D_m}\right)^{\frac{1}{b}}-1\right)E_1\left(\frac{\left(\frac{\sigma^2_s}{D_m}\right)^{\frac{1}{b}}-1}{q^*_2}\right)P(s)}=\bar{P}.
\end{split}
\end{equation}
\end{proposition}
\begin{proof}
The proof is provided in Appendix \ref{A1}.
\end{proof}

\begin{proposition}\label{P15}
For communication of a quasi-stationary source over a Rayleigh block
fading channel, the SCOPA-MDO scheme achieves an outage distortion exponent
$\Delta_{\textit{OD}}$ of the order $O\left({\frac{\bar{P}}{\ln{\bar{P}}}}\right)$ for large average power limit $\bar{P}$.
\end{proposition}
\begin{proof}
Noting \eqref{E102}, $\bar{P}\rightarrow \infty$ requires that
$q^*_2\rightarrow \infty$. In this case, using \eqref{E98}, we can
rewrite \eqref{E102} as follows.
\begin{equation}\label{E134}
\begin{split}
\underset{s:\sigma^2_s>
D_m}{\sum\!}{\!-\left(\left(\frac{\sigma^2_s}{D_m}\right)^{\frac{1}{b}}-1\right)\ln\frac{\left(\frac{\sigma^2_s}{D_m}\right)^{\frac{1}{b}}-1}{q^*_2}P(s)}\cong\bar{P}
\end{split}
\end{equation}
and after some manipulations, we obtain
\begin{equation}\label{E103}
\begin{split}
q^*_2\cong\exp\left(\frac{\bar{P}}{\underset{s:\sigma^2_s>
D_m}{\sum}{\left(\left(\frac{\sigma^2_s}{D_m}\right)^{\frac{1}{b}}-1\right)P(s)}}\right).\exp\!\left(\!\frac{\!\underset{s:\sigma^2_s>
D_m}{\sum}{\!\left(\left(\frac{\sigma^2_s}{D_m}\right)^{\frac{1}{b}}-1\right)\ln\left(\!\left(\frac{\sigma^2_s}{D_m}\!\right)^{\frac{1}{b}}-1\right)\!P(s)}\!}{\underset{s:\sigma^2_s>
D_m}{\sum}{\!\left(\left(\frac{\sigma^2_s}{D_m}\right)^{\frac{1}{b}}-1\right)\!P(s)}}\!\right).
\end{split}
\end{equation}
Considering \eqref{E46} and \eqref{E99} for $q^*_2\rightarrow
\infty$, we have
\begin{equation}\label{E104}
\text{P}_\text{Dout}\cong\frac{1}{q^*_2}\underset{s:\sigma^2_s>D_m}\sum{\left(\left(\frac{\sigma^2_s}{D_m}\right)^{\frac{1}{b}}-1\right)P(s)}.
\end{equation}
Hence, noting \eqref{E103} for $q^*_2\rightarrow \infty$ we obtain
\begin{equation}\label{E65}
\begin{split}
\Delta_{\textit{OD}}=\underset{\bar{P}\rightarrow\infty}{\lim}{\frac{\bar{P}}{\ln{\bar{P}}\underset{s:\sigma^2_s>
D_m}{\sum}{\left(\left(\frac{\sigma^2_s}{D_m}\right)^{\frac{1}{b}}-1\right)P(s)}}}.
\end{split}
\end{equation}
Thus, the proof is completed.
\end{proof}

The following Corollary expresses the implication of the SCOPA-MDO
design for transmission of a stationary source over block fading
channels. This is directly obtained from Proposition \ref{P7} when a
stationary source is assumed.
\begin{corollary}\label{P11}
The optimum channel coding rate and power adaptation prescribed by
SCOPA-MDO for transmission of a stationary source over a block
fading channel are given by
\begin{align}
R^*=\begin{cases}\frac{1}{2b}\log_2{\frac{\sigma^2}{D_m}}&\text{if }\alpha\geq\frac{\left(\frac{\sigma^2}{D_m}\right)^\frac{1}{b}-1}{q^*_2}\\
0&\text{otherwise}
\end{cases}
\end{align}
\begin{align}\label{E15}
\gamma^*=\begin{cases}\frac{\left(\frac{\sigma^2}{D_m}\right)^\frac{1}{b}-1}{q^*_2}&\text{if
}\alpha\geq\frac{\left(\frac{\sigma^2}{D_m}\right)^\frac{1}{b}-1}{q^*_2},
\\0&\text{otherwise}\end{cases}
\end{align}
in which $q^*_2$ is selected such that the power constraint is
satisfied with equality, i.e.,
\begin{equation}\label{E17}
\underset{\alpha\geq\frac{\left(\frac{\sigma^2}{D_m}\right)^\frac{1}{b}-1}{q^*_2}}{\int}{\frac{\left(\frac{\sigma^2}{D_m}\right)^\frac{1}{b}-1}{\alpha}
f(\alpha)\,\mathrm{d}\alpha}=\bar{P}.
\end{equation}
Specifically for Rayleigh block fading channel, \eqref{E17} is
simplified to
\begin{equation}\label{E18}
\left(\left(\frac{\sigma^2}{D_m}\right)^\frac{1}{b}-1\right)E_1\left(\frac{\left(\frac{\sigma^2}{D_m}\right)^\frac{1}{b}-1}{q^*_2}\right)=\bar{P}.
\end{equation}
\end{corollary}

\begin{remark}\label{rem1}
Examining Corollaries \ref{P10} and \ref{P11} for transmission of a
stationary source over a block fading channel, one sees that the
SCOPA-MDO and COPA-MDO schemes provide the same distortion outage
probability and outage distortion exponent. This implies that in
this setting with power adaptation an optimized fixed rate provides
all the gain that may be obtained by rate adaptation.
\end{remark}

The performance of SCOPA-MDO scheme is studied in Section \ref{SV}.


\section{Performance Evaluation}\label{SV}
In this section, we first present two constant power transmission
schemes as benchmarks for comparisons. Next, we consider analytical
performance comparison of different schemes followed by numerical
results. To this end, we consider four quasi-stationary sources,
with $N_s=25$ where the variance of the source in the state $s$ is
given by
$\sigma^2_s(s)=(1+\frac{1}{6}s)^2:\forall{s}\in\{0,1,...,N_S-1\}$.
For three of the sources, labeled as G1, G2 and G3, the probability
of being in different states follows a discrete Gaussian
distribution with mean 3 and variances 0.05, 0.48, 1.07,
respectively. For the fourth source, U, the said distribution is
considered uniform with the same mean and a variance of 1.44. We
also consider an stationary source S with $\sigma_s=3$ for a
meaningful comparison. Unless otherwise mentioned, we consider the
source G2 for the following results and simulations.

\subsection{Benchmark}
Two constant power schemes for transmission of a quasi-stationary
source over a block fading channel are considered as benchmarks for
comparisons. In the first scheme, the channel coding rate is
adjusted based on the channel state to minimize the distortion outage probability; hence
the scheme is labeled as Channel Optimized Rate Adaptation with
Constant Power (CORACP). In the second scheme with Constant Rate and
Constant Power (CRCP), the aim is to find the optimized fixed rate such that the distortion outage probability is minimized.

\subsubsection{Channel Optimized Rate Adaptation with Constant Power}
With CORACP and constant transmission power $\bar{P}$, the
instantaneous capacity is given by
$C=\frac{1}{2}\log_2{\left(1+\alpha\bar{P}\right)}$; and hence to
minimize $\text{P}_{\text{Dout}}$ it is logical to consider the rate
adaptation strategy of $R=C$. The source coding rate is then set as
$R_s=bR$. The next two Propositions quantify the distortion outage
performance of CORACP.

\begin{proposition}\label{P1}
The distortion outage probability for transmission of a
quasi-stationary source over a Rayleigh block fading channel using
CORACP is given by
\begin{equation}\label{E5}
\begin{split}
\text{P}_\text{Dout}=\text{Pr}(\sigma^2_s>D_m)-\underset{s:\sigma^2_s>D_m}\sum\exp\left(-\frac{1}{\bar{P}}\left[\left(\frac{\sigma^2_s}{D_m}\right)^{\frac{1}{b}}-1\right]\right)P(s){.}
\end{split}
\end{equation}
\end{proposition}
\begin{proof}
The proof is provided in Appendix \ref{A4}.
\end{proof}

\begin{proposition}\label{P12}
For communication of a quasi-stationary source over a Rayleigh block
fading channel, the CORACP scheme achieves the outage distortion exponent $\Delta_{\textit{OD}}$ of the order $O\left(1\right)$.
\end{proposition}

\begin{proof}
Using \eqref{E99} and considering large enough $\bar{P}$, \eqref{E5}
can be rewritten as
\begin{equation}\label{E11}
\text{P}_\text{Dout}\cong\frac{1}{\bar{P}}\underset{s:\sigma^2_s>D_m}\sum\left(\left(\frac{\sigma^2_s}{D_m}\right)^{\frac{1}{b}}-1\right)P(s).
\end{equation}
Thus,
\begin{equation}\label{E29}
\begin{split}
\Delta_{\textit{OD}}=\underset{\bar{P}\rightarrow\infty}{\lim}{-\frac{-\ln{\bar{P}}+\ln{\underset{s:\sigma^2_s>D_m}\sum\left(\left(\frac{\sigma^2_s}{D_m}\right)^{\frac{1}{b}}-1\right)P(s)}}{\ln{\bar{P}}}}=1
\end{split}
\end{equation}
and the proof is complete.
\end{proof}
Based on Propositions \ref{P1} and \ref{P12}, the following
Corollary presents the performance of CORACP with stationary
sources.
\begin{corollary}\label{P23}
For communication of a stationary source over a Rayleigh block
fading channel, the CORACP scheme achieves a distortion outage
probability of
\begin{align}\label{E58}
\text{P}_\text{Dout}=1-\exp\left(-\frac{\left(\frac{\sigma^2}{D_m}\right)^\frac{1}{b}-1}{\bar{P}}\right)
\end{align}
and an outage distortion exponent of
$\Delta_{\textit{OD}}$ of the order $O\left(1\right)$.
\end{corollary}

\subsubsection{Constant Rate Constant Power}
With CRCP, the fixed rate is optimized to minimize the distortion outage
probability when the power is constant. The following three
Propositions express the optimized fixed rate, distortion outage
probability and distortion outage exponent obtained by CRCP.

\begin{proposition}\label{P25}
For transmission of a quasi-stationary source over a block fading
channel using the CRCP scheme, the optimum rate, $R^*$, is an
element of the set
$\{\frac{1}{2b}\log_2\frac{\sigma^2_s}{D_m};s\in\mathcal{S},\;\sigma^2_s\geq
D_m\}$ that maximizes
\begin{equation}
{\text{Pr}\left(\alpha\geq\frac{2^{2R}-1}{\bar{P}}\right)\biggl(\text{Pr}\left(\!\sigma^2_s>D_m\right)-\text{Pr}\left(\sigma^2_s2^{-2bR}>D_m\right)\!\biggr)}\!
\end{equation}
\end{proposition}
\begin{proof}
The proof is similar to the proof of Proposition \ref{PR6}.
\end{proof}
\begin{proposition}\label{P26}
For transmission of a quasi-stationary source over a Rayleigh block
fading channel, the CRCP scheme achieves the distortion outage
probability of
\begin{align}
\text{P}_\text{Dout}=\left(1-\exp\left(-\frac{2^{2R^*}-1}{\bar{P}}\right)\right)\text{Pr}(\sigma^2_s>D_m)
+\exp\left(-\frac{2^{2R^*}-1}{\bar{P}}\right)\text{Pr}(\sigma^2_s2^{-2bR^*}>D_m),
\end{align}
with $R^*$ in Proposition \ref{P25}, and the outage distortion exponent $\Delta_{\textit{OD}}$ of the order $O\left(1\right)$.
\end{proposition}
\begin{proof}
Noting \eqref{E75} and constant power $\bar{P}$ with exponentially
distributed channel gain, obtaining $\text{P}_\text{Dout}$ is
straightforward. For a large power constraint, one can verify that
$R^*$ is simplified to
\begin{equation}\label{E42}
R^*=\frac{1}{2b}\log_2\frac{\underset{s}{\operatorname{max}}\{\sigma^2_s\}}{D_m}.
\end{equation}
Thus, from Proposition \ref{P26} and noting \eqref{E99} and \eqref{E42}, we obtain
\begin{equation}\label{E44}
\begin{split}
\text{P}_\text{Dout}\cong
\frac{\left(\frac{\underset{s}{\operatorname{max}}\{\sigma^2_s\}}{D_m}\right)^{\frac{1}{b}}-1}{\bar{P}}\text{Pr}\left(\sigma^2_s>D_m\right)
\end{split}
\end{equation}
and therefore,
\begin{equation}\label{E26}
\begin{split}
\Delta_{\textit{OD}}=\underset{\bar{P}\rightarrow\infty}{\lim}{-\frac{-\ln(\bar{P})+\ln\left(\left({\left(\frac{\underset{s}{\operatorname{max}}\{\sigma^2_s\}}{D_m}\right)^{\frac{1}{b}}-1}\right)\text{Pr}\left(\sigma^2_s>D_m\right)\right)}{\ln{\bar{P}}}}=1.
\end{split}
\end{equation}
\end{proof}

Based on the above two Propositions, the following two Corollaries present the optimum rate and the performance of CRCP with stationary sources.
\begin{corollary}\label{P28}
The optimum channel coding rate prescribed by CRCP for transmission
of a stationary source over a block fading channel is given by
\begin{align}\label{E8}
R^*=\frac{1}{2b}\log_2{\frac{\sigma^2}{D_m}}
\end{align}
\end{corollary}
\begin{corollary}\label{P29}
For communication of a stationary source over a Rayleigh block fading channel, the CRCP scheme achieves a distortion outage probability of
\begin{align}\label{E58}
\text{P}_\text{Dout}=1-\exp\left(-\frac{\left(\frac{\sigma^2}{D_m}\right)^\frac{1}{b}-1}{\bar{P}}\right)
\end{align}
and an outage distortion exponent of
$\Delta_{\textit{OD}}$ of the order $O\left(1\right)$.
\end{corollary}
\begin{remark}\label{rem2}
Noting Corollaries \ref{P23} and \ref{P29} for transmission of a
stationary source over a block fading channel, one sees that the
CORACP and CRCP schemes provide the same distortion outage
probability and outage distortion exponent. This implies that in
this setting with constant power, an optimized fixed rate provides
all the gain that may be obtained by rate adaptation.
\end{remark}


\subsection{Analytical Performance Comparison}\label{SVB}
In the sequel, we quantify the respective asymptotic outage
distortion gain $G_{\textit{OD}}$ of SCOPA-MDO, COPA-MDO, CORACP and
CRCP for transmission of a quasi-stationary source over a block
fading channel.
\begin{proposition}\label{P17}
In transmission of a quasi-stationary source over a Rayleigh block
fading channel, the asymptotic outage distortion gain obtained by
SCOPA-MDO with respect to COPA-MDO is given by
\begin{equation}\label{E1}
G\!_{O\!D}=10\log_{10}\left(\left(\frac{\underset{s}{\operatorname{max}}\{\sigma^2_s\}}{D_m}\right)^{\frac{1}{b}}-1\right)-10\log_{10}{\underset{s:\sigma^2_s>
D_m}{\sum}{\left(\left(\frac{\sigma^2_s}{D_m}\right)^{\frac{1}{b}}-1\right)P(s)}}.
\end{equation}

\end{proposition}
\begin{proof}
The proof is provided in Appendix \ref{A2}.
\end{proof}
\begin{proposition}\label{P18}
In transmission of a quasi-stationary source over a Rayleigh block
fading channel, the asymptotic outage distortion gain obtained by
COPA-MDO with respect to CORACP is given by
\begin{equation}\label{E2}
G\!_{O\!D}=10\log_{10}{\frac{\bar{P}_2}{\left(\frac{\underset{s}{\operatorname{max}}\{\sigma^2_s\}}{D_m}\right)^{\frac{1}{b}}-1}}-10\log_{10}{\ln{\frac{\bar{P}_2}{\underset{s:\sigma^2_s>
D_m}{\sum}{\left(\left(\frac{\sigma^2_s}{D_m}\right)^{\frac{1}{b}}-1\right)P(s)}}}},
\end{equation}
where $\bar{P}_1$ and $\bar{P}_2$ are power limits in COPA-MDO and
CORACP; $\bar{P}_1(\rm{dB})=$$\bar{P}_2(\rm{dB})$-$G\!_{O\!D}$.
\end{proposition}
\begin{proof}
The proof is provided in Appendix \ref{A2}.
\end{proof}

\begin{proposition}\label{P20}
In transmission of a quasi-stationary source over a Rayleigh block
fading channel, the asymptotic outage distortion gain obtained by
CORACP with respect to CRCP is given by
\begin{equation}\label{E47}
G\!_{O\!D}=10\log_{10}\left(\left(\left(\frac{\underset{s}{\operatorname{max}}\{\sigma^2_s\}}{D_m}\right)^{\frac{1}{b}}-1\right)\text{Pr}\left(\sigma^2_s>D_m\right)\right)-10\log_{10}{\underset{s:\sigma^2_s>
D_m}{\sum}{\left(\left(\frac{\sigma^2_s}{D_m}\right)^{\frac{1}{b}}-1\right)P(s)}}.
\end{equation}
\end{proposition}
\begin{proof}
The proof is provided in Appendix \ref{A2}.
\end{proof}

As evident, $G_{\textit{OD}}$ of SCOPA-MDO with respect to
COPA-MDO and CORACP with respect to CRCP are
independent of the power. This is while the gain of COPA-MDO with respect to
CORACP depends nonlinearly on the average transmission power limit
and it improves when the power limit increases. The dependency of
$G_{\textit{OD}}$ on the variance of the source in different states
and the maximum allowable distortion is also seen in \eqref{E1} to
\eqref{E47}. Table \ref{T4} presents the value of $G\!_{O\!D}$ at $D_m=8 \rm dB$.

The distortion outage exponent $\Delta_{\textit{OD}}$ of the
COPA-MDO and SCOPA-MDO schemes which are derived in line with the
proofs of the Propositions \ref{P13} and \ref{P15}, are quantified
in Table \ref{T3}. As denoted in Propositions \ref{P12} and
\ref{P26}, CORACP and CRCP schemes give $\Delta_{\textit{OD}}=1$.
The distortion outage exponent indicates the speed at which the
distortion outage ($\rm dB$) reduces as the average power (limit)
($\rm dB$) increases. Therefore, as evident, this speed is
noticeably high with SCOPA-MDO and very low with CORACP and CRCP.
Furthermore, the $\Delta_{\textit{OD}}$ obtained by SCOPA-MDO and
COPA-MDO depends on the average power limit $\bar{P}$, maximum
allowable distortion $D_m$ and bandwidth expansion ratio $b$. It is
observed that with SCOPA-MDO and COPA-MDO, $\Delta_{\textit{OD}}$
improves as $b$, $\bar{P}$ or $D_m$ increase. In fact, for a given
value of $N$, a larger $b$ implies a more finely encoded source that
is more sensitive to channel errors and hence can more greatly
benefit from increased power. The results in Tables I and II
indicate that from the perspective of probability of distortion
outage, for delay-limited communication of quasi-stationary sources,
CORACP and CRCP schemes may not be appropriate designs.

\begin{table}[!t]
\caption{Asymptotic outage distortion gain $G\!_{O\!D}$ of scheme 1
with respect to scheme 2 for source G2 evaluated at $b=1$ and $D_m=8
\rm dB$.} \label{T4} \centering
\begin{tabular}{|c|c|c|c|}
\hline Scheme 1&Scheme 2&$G\!_{O\!D}$ at $\bar{P}_2=25 \rm dB$&$G\!_{O\!D}$ at $\bar{P}_2=20 \rm dB$\\
\hline SCOPA-MDO&COPA-MDO&7.14&7.14\\
\hline
COPA-MDO&CORACP&12.28&8.16\\
\hline
CORACP&CRCP&5.74&5.74\\
\hline
\end{tabular}
\end{table}
\begin{table*}[t]
\caption{Distortion outage exponent $\Delta_{\textit{OD}}$ of the
proposed schemes for source G2}\label{T3} \centering
\begin{tabular}{|c|c|c||c|c|}
\hline
Bandwidth Expansion Ratio&\multicolumn{2}{c||}{$b$=1} &\multicolumn{2}{c|}{$b$=5}\\
\hline
Schemes and Settings&\tiny{SCOPA-MDO} & \tiny{COPA-MDO}& \tiny{SCOPA-MDO} & \tiny{COPA-MDO}  \\
\hline $\tiny{D_m}=8\rm dB$,  $\tiny{\bar{P}=16\rm dB}$  &18.90& 3.89&129.08&34.33\\
\hline $\tiny{D_m=8 \rm dB}$, $\tiny{\bar{P}=20\rm dB}$  &37.97&7.53&259.39&68.69\\
\hline $\tiny{D_m=5\rm dB}$,  $\tiny{\bar{P}=20\rm dB}$  &10.80&3.27& 95.75&42.53\\
\hline
\end{tabular}
\end{table*}
The following three Corollaries quantify the asymptotic outage
distortion gain in transmission of a stationary source over block
fading channels. These are directly obtained from Propositions
\ref{P17} and \ref{P20} when a stationary source is considered.

\begin{corollary}\label{P24}
In transmission of a stationary source over a Rayleigh block fading
channel, the asymptotic outage distortion gain obtained by
SCOPA-MDO with respect to COPA-MDO is equal to zero.
\end{corollary}

\begin{corollary}\label{P22}
In transmission of a stationary source over a Rayleigh block fading
channel, the asymptotic outage distortion gain of COPA-MDO with
respect to CORACP is equal to
\begin{equation}\label{E38}
G\!_{O\!D}=10\log_{10}{\frac{X}{\ln{X}}},
\end{equation}
where
$X:=\frac{\bar{P}_2}{\left(\frac{\sigma^2}{D_m}\right)^{\frac{1}{b}}-1}$
and $\bar{P}_2$ is the power limit of CORACP.
\end{corollary}

\begin{corollary}\label{P24}
In transmission of a stationary source over a Rayleigh block fading
channel, the asymptotic outage distortion gain of CORACP with respect to CRCP is equal to zero.
\end{corollary}

The values of $G\!_{O\!D}$ in \eqref{E38} and $\Delta_{\textit{OD}}$
in Corollary \ref{P21} for the stationary source S, $\bar{P}_2=20\rm
dB$ and $D_m=8\rm dB$ respectively amount to $16.33 \rm dB$ and
$51$. Noting Remarks \ref{rem1} and \ref{rem2}, it is observed that
with stationary sources and common settings, an optimized fixed rate
scheme and an adaptive rate scheme provide the same distortion
outage probability. Therefore, in this setting, with optimum power
adaptation, COPA-MDO and SCOPA-MDO schemes; and with constant power,
CORACP and CRCP schemes provide the same distortion outage
probability.



\subsection{Numerical Results}\label{SVC}
Figs. 2a and 2b depict the distortion outage probability performance
of the presented schemes as a function of the power constraint
$\bar{P}$ for $D_m=\rm 8\, dB$ and $D_m=\rm 5\, dB$, respectively.
As expected, for a given $\bar{P}$, $\text{P}_\text{Dout}$ decreases
as $D_m$ increases. It is evident that the proposed SCOPA-MDO scheme
achieves an asymptotic outage distortion gain of about $\rm7.1\, dB$
and $\rm5.6\,  dB$ with respect to COPA-MDO, for $\bar{P}=\rm 20\,
dB$ and $D_m=\rm 8\, dB$ and $D_m=\rm 5\,dB$, respectively. In the
same settings, the COPA-MDO scheme achieves asymptotic outage
distortion gains of about $\rm9.1\, dB$ and $\rm6.3\, dB$ with
respect to CORACP; and CORACP achieves gains of $\rm6\, dB$ and
$\rm4.9 \, dB$ with respect to CRCP. The results obtained from
simulations and what is reported in Table \ref{T4} from analyses
match reasonably well given the assumption of very high average SNR
considered in the analytical performance evaluations.

The analytical results in Table \ref{T3} for $\Delta_{\textit{OD}}$
performance, may also be observed in numerical results of Figs. 2a
and 2b. Specifically, at each point on the curves, the corresponding
value in the vertical coordinate  in $\rm dB$, i.e.,
$10\log_{10}{\text{P}_\text{Dout}}$ divided by the value in the
horizontal coordinate, i.e., $\bar{P}(\rm dB)$, indicates
$-\Delta_{\textit{OD}}$. For example, as seen in Fig. 2b,
$\Delta_{\textit{OD}}$ for SCOPA-MDO is almost equal to $38$ at
$\bar{P}=20\rm \,dB$ with $D_m=8\rm \,dB$.

Figs. 3a and 3b depict the $\text{P}_\text{Dout}$ performance of the
presented schemes for different sources. As observed, as the
non-stationary characteristics of the source intensifies  (from
source S to U), the distortion outage probability increases in
general. However, the sensitivity of the performance of different
schemes to the level of non-stationarity varies. Once again, the
advantage of power adaptation is clear. For the stationary source S,
as noted in Remarks \ref{rem1} and \ref{rem2} and also seen in these
Figures, SCOPA-MDO and COPA-MDO schemes provide the same distortion
outage probability in this setting. This also holds true for CORACP
and CRCP. It is observed that SCOPA-MDO (or equivalently COPA-MDO)
scheme achieves an asymptotic outage distortion gain of about
$\rm16.8\, dB$ with respect to CORACP (or equivalently CRCP) for
$\bar{P}=\rm 20\, dB$ and $D_m=\rm 8\, dB$.

It is noteworthy that the four methods discussed, rely on different
levels of source and channel state information (SSI and CSI).
Specifically, it can be verified that three schemes of SCOPA-MDO,
COPA-MDO and CORACP require instantaneous CSI for rate and/or power
adaptation, while CRCP needs CSI statistics. The SCOPA-MDO scheme also needs instantaneous SSI, while
COPA-MDO and CRCP simply need SSI statistics.

\section{Conclusions}\label{SVI}
In this paper, delay-limited transmission of a quasi-stationary
source over a block fading channel was considered. Aiming at
minimizing the distortion outage probability, two transmission
strategies namely channel-optimized power adaptation with fixed rate
(COPA-MDO) and source and channel optimized power (and rate)
adaptation (SCOPA-MDO) were introduced. The SCOPA-MDO scheme
provides a superior performance, while the COPA-MDO scheme enjoys
the simplicity of single rate transmission. In high SNR regime,
different scaling laws involving outage distortion exponent and
asymptotic outage distortion were derived. Our studies confirm the
benefit of power adaption from a distortion outage perspective and
for delay-limited transmission of quasi-stationary sources over
wireless block fading channels. The analyses of the presented
schemes in the case of stationary sources indicate the same outage
distortion performance with or without rate adaptation.




\appendices

\section{Proof of Proposition \ref{P8}}\label{A1}
Noting \eqref{E19} and \eqref{E95} for exponentially distributed
channel gain, we can obtain \eqref{E102} and subsequently $q^*_2$ by
numerical methods. Using \eqref{E14} and \eqref{E24}, the distortion
for each state of the source and the channel is given by
\begin{equation}\label{E36}
D(\sigma_s,\alpha,\gamma)=\begin{cases} D_m  &\text{if}
\;\frac{\sigma^2_s}{D_m}> 1
\;\text{and}\;\frac{\left[\left(\frac{\sigma^2_s}{D_m}\right)^{\frac{1}{b}}-1\right]}{\alpha}<
q^*_2\\
\sigma^2_s &\text{otherwise}{.}
\end{cases}
\end{equation}
Therefore, we have
\begin{equation}
\text{P}_\text{Dout}=\text{Pr}\left(\frac{1}{\alpha}\left[\left(\frac{\sigma^2_s}{D_m}\right)^{\frac{1}{b}}-1\right]\geq
q^*_2,\,\sigma^2_s>D_m\right).
\end{equation}
Considering exponentially distributed channel gain $\alpha$,
deriving \eqref{E46} is straightforward.

\section{Proof of Proposition \ref{P1}}\label{A4}
Noting \eqref{E24} for $\gamma=\bar{P}$, distortion can be written
as follows
\begin{equation}\label{E4}
D(\sigma_s,\alpha,\gamma=\bar{P})=\sigma^2_s
2^{-2bR}=\frac{\sigma^2_s}{\left(1+\alpha \bar{P}\right)^b}
\end{equation}
and then we can derive
\begin{equation}\label{E6}
\begin{split}
\text{P}_\text{Dout}&=\text{Pr}\left(\frac{\sigma^2_s}{\left(1+\alpha
\bar{P}\right)^b}>D_m,\sigma^2_s>D_m\right)=\text{Pr}\left(\alpha<\frac{1}{\bar{P}}\left(\left(\frac{\sigma^2_s}{D_m}\right)^{\frac{1}{b}}-1\right),\sigma^2_s>D_m\right).
\end{split}
\end{equation}
Considering exponentially distributed channel gain $\alpha$,
obtaining \eqref{E5} is straightforward.

\section{Proof of Propositions \ref{P17} and \ref{P18}}\label{A2}
The average power to asymptotically achieve a certain distortion
outage probability using SCOPA-MDO and COPA-MDO schemes are denoted
by $\bar{P}_1$ and $\bar{P}_2$, respectively. Thus, we can use
\eqref{E68} to derive  $G\!_{O\!D}$. Noting \eqref{E103},
\eqref{E104} and \eqref{E69} we set
\begin{equation}\label{E71}
\begin{split}
\underset{\bar{P_1}\rightarrow\infty}{\lim}\frac{\bar{P}_1}{\underset{s:\sigma^2_s>
D_m}{\sum}{\left(\left(\frac{\sigma^2_s}{D_m}\right)^{\frac{1}{b}}-1\right)P(s)}}=\underset{\bar{P_2}\rightarrow\infty}{\lim}\frac{\bar{P}_2}{\left(\frac{\underset{s}{\operatorname{max}}\{\sigma^2_s\}}{D_m}\right)^{\frac{1}{b}}-1}.
\end{split}
\end{equation}
Therefore, we can derive \eqref{E1} and complete the proof of
Proposition \ref{P17}.

The proof of Propositions \ref{P18} or \ref{P20} is straightforward,
when we use \eqref{E69} and \eqref{E11} or \eqref{E11} and
\eqref{E44}; and obtain the following
\begin{equation}\label{E45}
\begin{split}
\underset{\bar{P_2}\rightarrow\infty}{\lim}\frac{\bar{P}_1}{\left(\frac{\underset{s}{\operatorname{max}}\{\sigma^2_s\}}{D_m}\right)^{\frac{1}{b}}-1}=\underset{\bar{P}_2\rightarrow\infty}{\lim}\ln{\frac{\bar{P}_2}{\underset{s:\sigma^2_s>
D_m}{\sum}{\left(\left(\frac{\sigma^2_s}{D_m}\right)^{\frac{1}{b}}-1\right)P(s)}}}
\end{split}
\end{equation}

\begin{align}\label{E77}
\underset{\bar{P}_2\rightarrow\infty}{\lim}{\frac{\bar{P}_1}{\underset{s:\sigma^2_s>
D_m}{\sum}{\left(\left(\frac{\sigma^2_s}{D_m}\right)^{\frac{1}{b}}-1\right)P(s)}}}=\underset{\bar{P_2}\rightarrow\infty}{\lim}{\frac{\bar{P}_2}{\left(\left(\frac{\underset{s}{\operatorname{max}}\{\sigma^2_s\}}{D_m}\right)^{\frac{1}{b}}-1\right)\text{Pr}\left(\sigma^2_s>D_m\right)}}.
\end{align}

\ifCLASSOPTIONcaptionsoff
  \newpage
\fi



\bibliographystyle{Ieeetr}
\bibliography{IEEEabrv,Ref}

%
%
%


\begin{figure}[h!]
\centering
\includegraphics[width=3.6in]{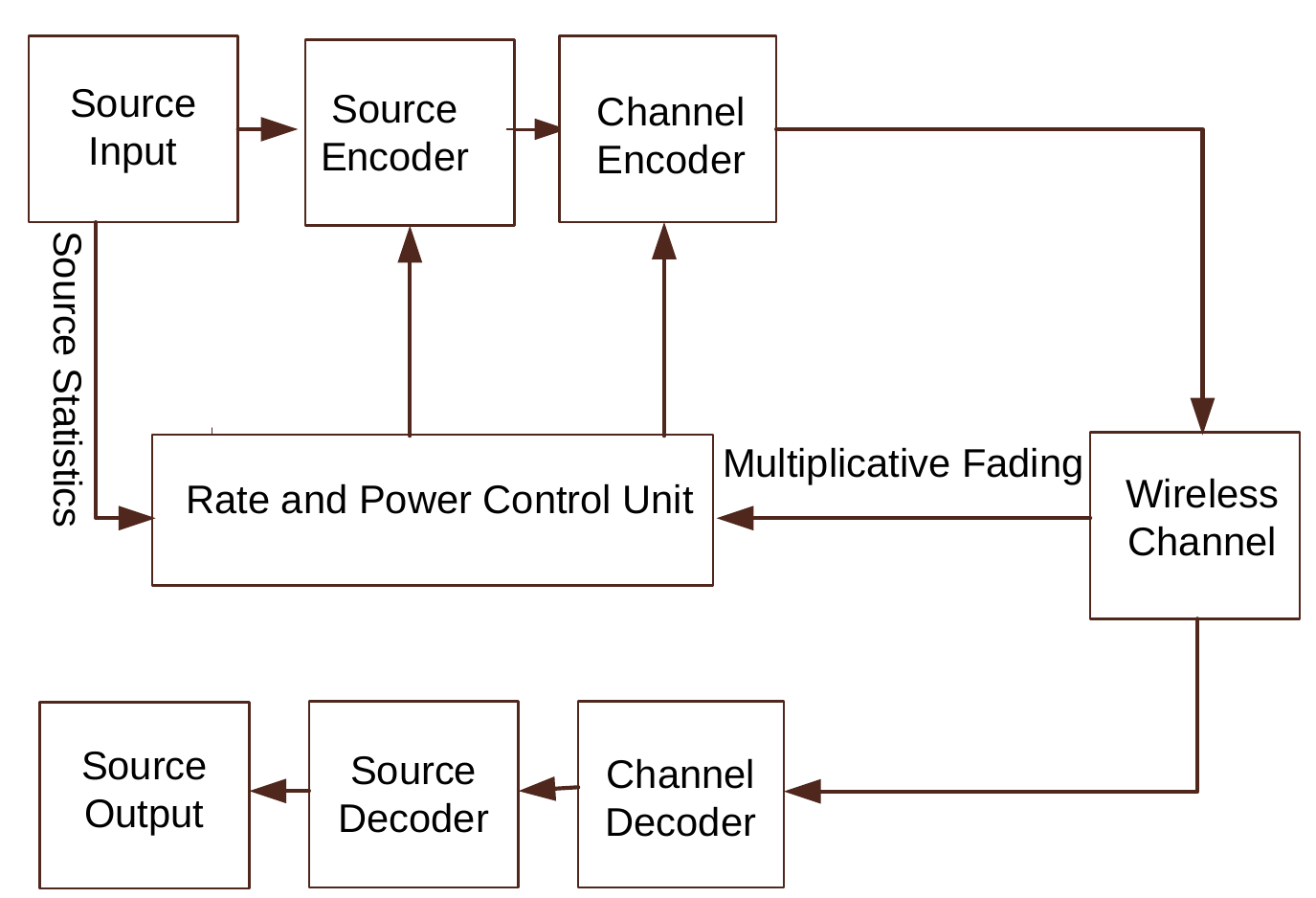}
\caption{Block diagram of the system.} \label{F20}
\end{figure}


\begin{figure}[h!]
\centering
\begin{tabular}{cc}
\includegraphics[width=3.1in]{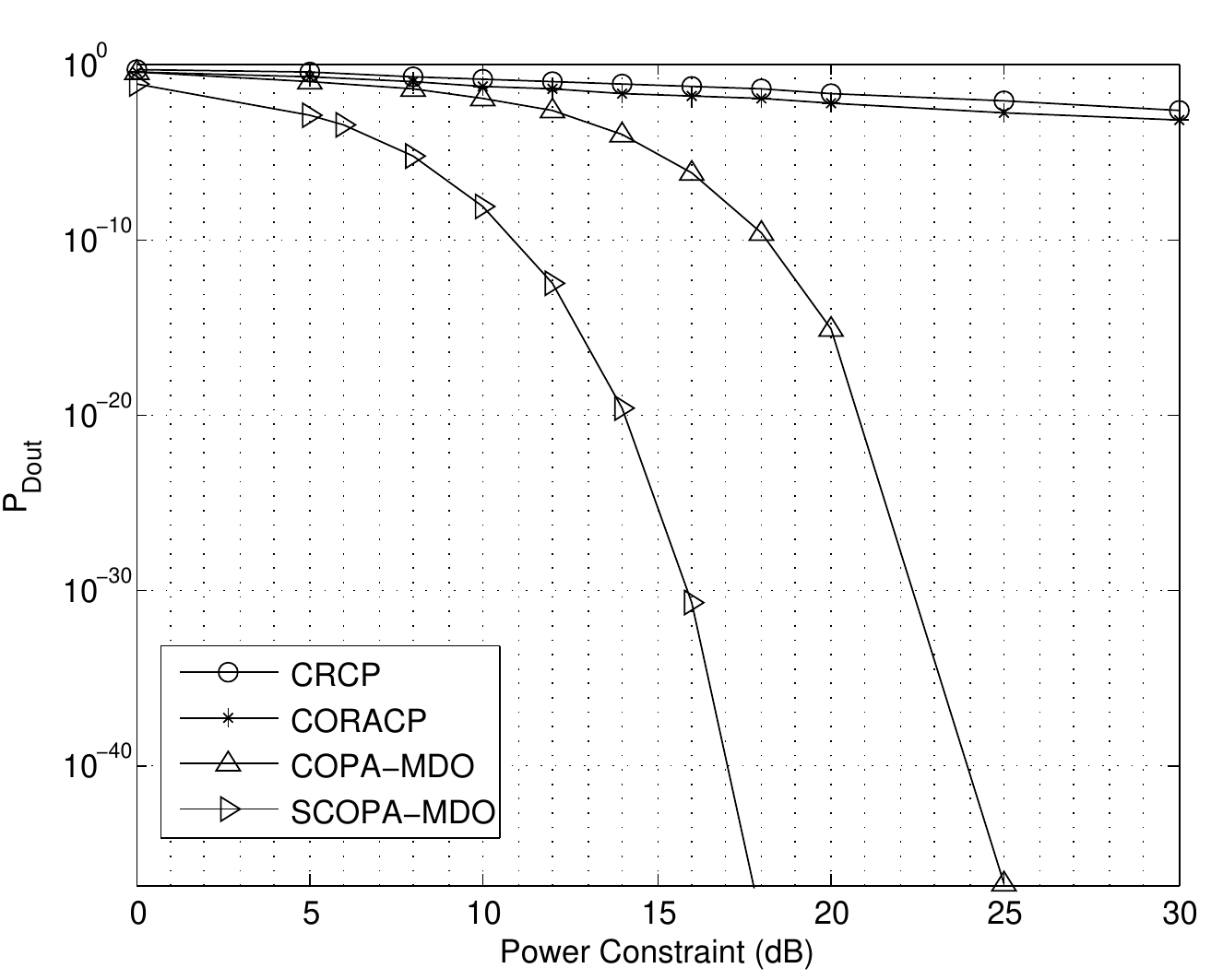}&
\includegraphics[width=3.1in]{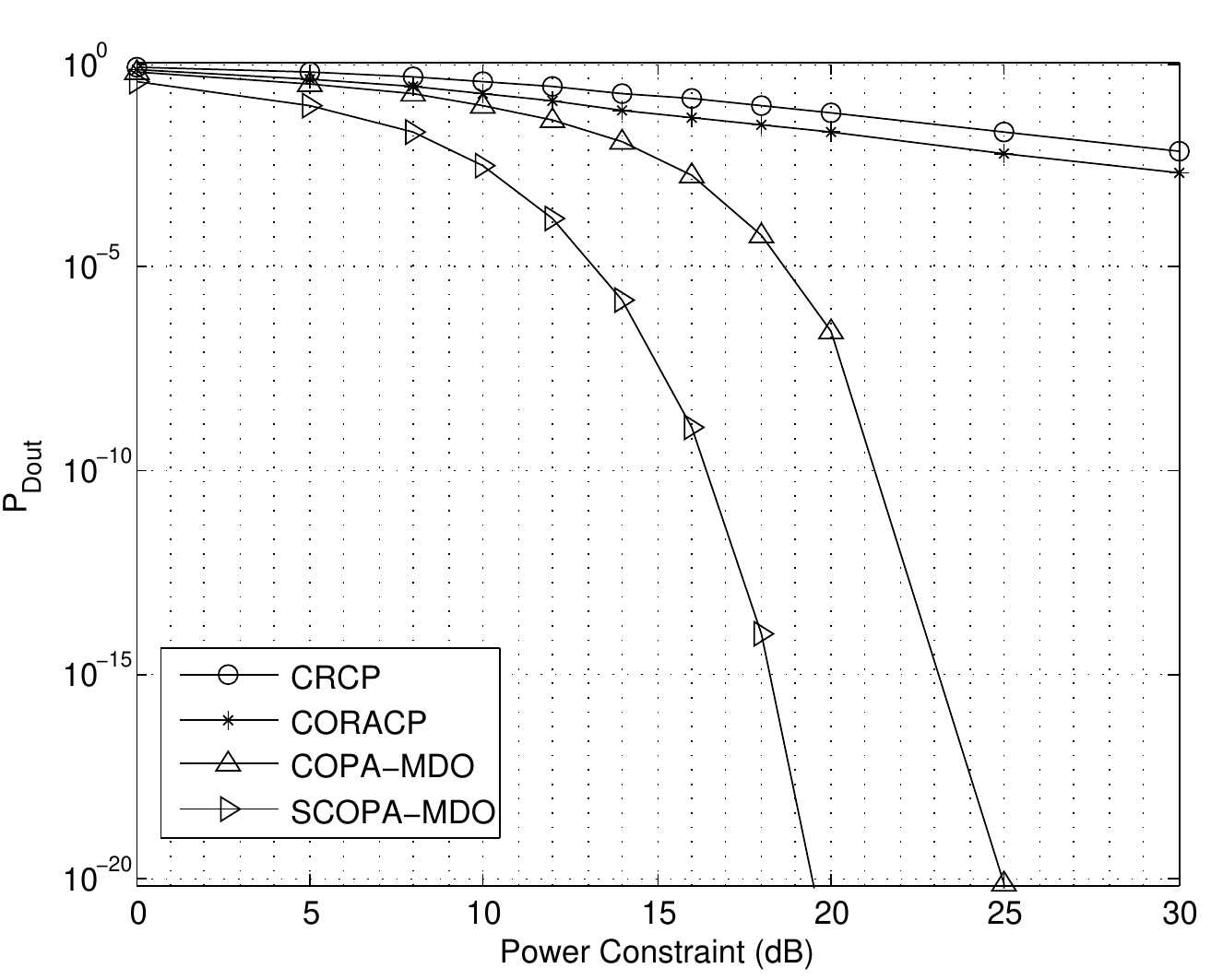}\\
\scriptsize(a)&\scriptsize(b)\\
\end{tabular}
\caption{Distortion outage probability versus $\bar{P}$; $b=1$
and (a) $D_m=\rm8\,dB$ (b)$D_m=\rm5\,dB$.}
\end{figure}

\begin{figure}[h!]
\centering
\begin{tabular}{cc}
\includegraphics[width=3.1in]{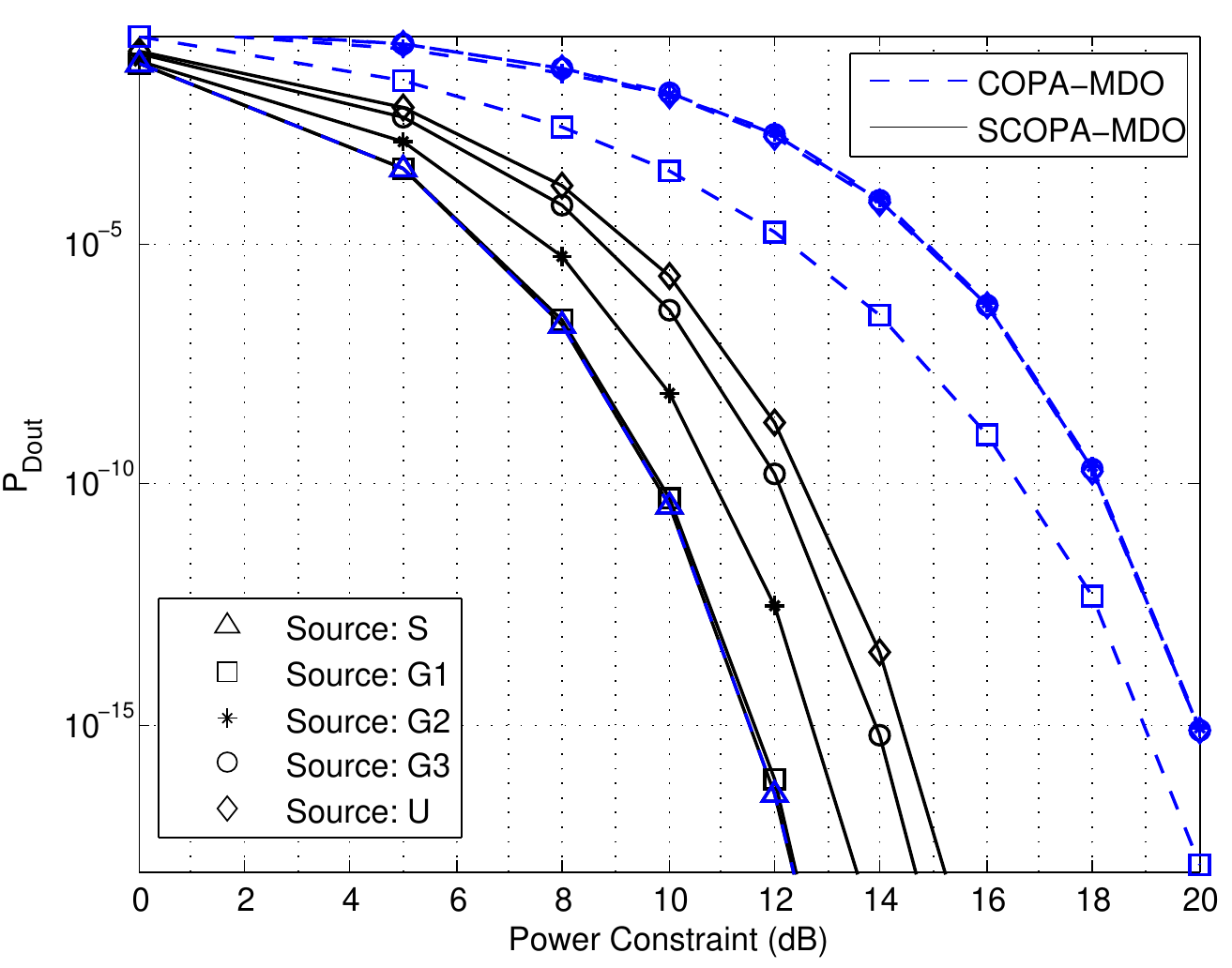}&
\includegraphics[width=3.1in]{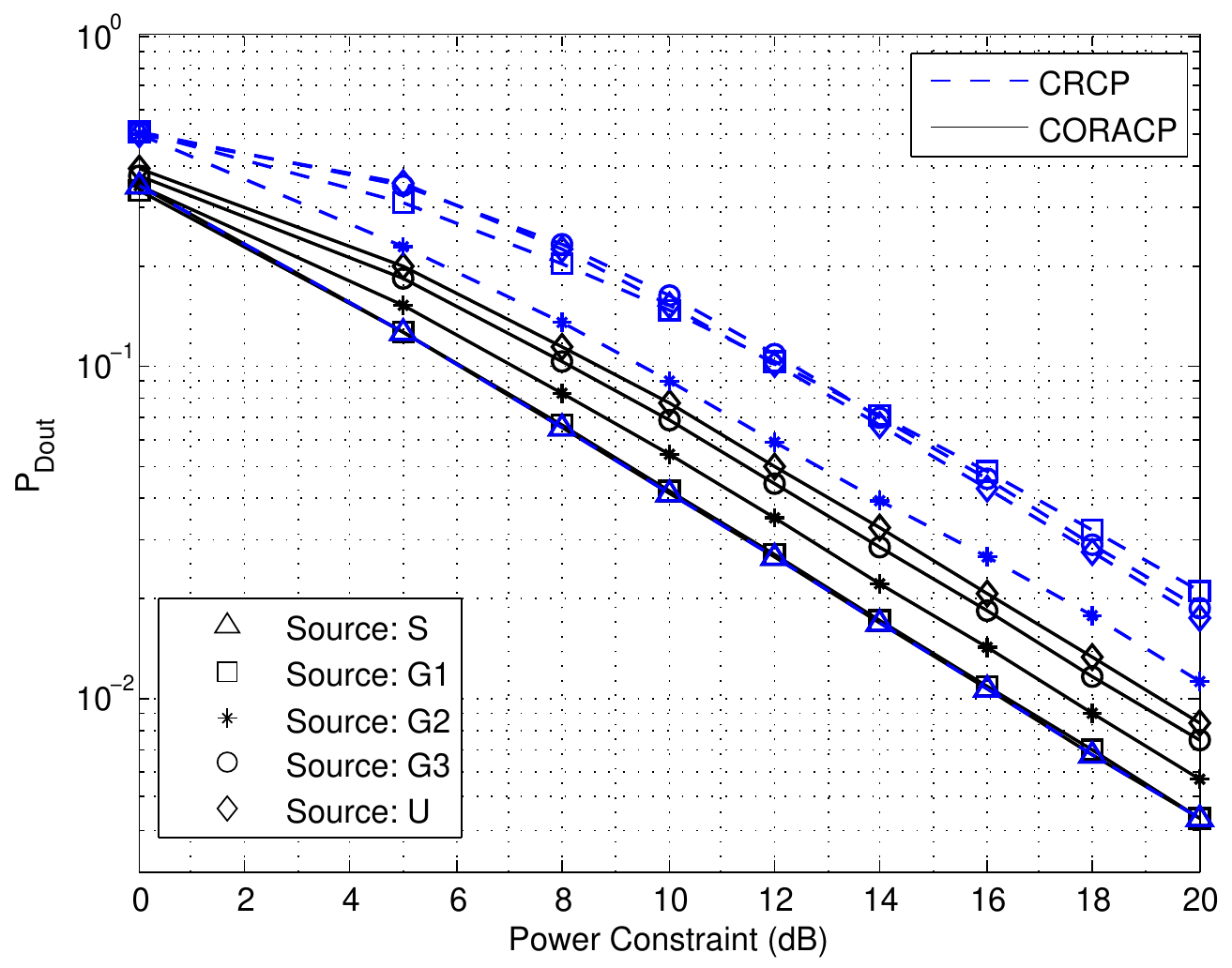}\\
\scriptsize(a)&\scriptsize(b)\\
\end{tabular}
\caption{Distortion outage probability versus $\bar{P}$ for five
different sources; $b=1$ and $D_m=\rm8\,dB$.} \label{F6}
\end{figure}
\!

\end{document}